\documentclass[twocolumn,letterpaper,10pt,conference]{ieeeconf}  

\IEEEoverridecommandlockouts                              
\usepackage{mathptmx} 
\usepackage{times} 
\usepackage{amsmath} 
\usepackage{amssymb}  

\usepackage[noadjust,sort]{cite}
\usepackage{accents}
\newcommand{\ubar}[1]{\underaccent{\bar}{#1}}

\title{\LARGE \bf Efficient Containment of Exact SIR Markovian Processes on Networks}

\author{Masaki Ogura and Victor M.~Preciado
\thanks{The authors are with the Department of Electrical and Systems
Engineering, University of Pennsylvania, Philadelphia, PA 19014, USA.
Email:  {\tt\small \{ogura,preciado\}@seas.upenn.edu}}%
\thanks{This work was supported in part by the NSF under grants CNS-1302222 and IIS-1447470.}%
}

\newtheorem{definition}{Definition}[section]
\newtheorem{assumption}[definition]{Assumption}
\newtheorem{lemma}[definition]{Lemma}
\newtheorem{proposition}[definition]{Proposition}
\newtheorem{theorem}[definition]{Theorem}
\newtheorem{remark}[definition]{Remark}

\newtheorem{problem}[definition]{Problem}

\DeclareMathOperator*{\minimize}{minimize}
\DeclareMathOperator*{\subjectto}{subject\ to}
\DeclareMathOperator{\Id}{Id}

\newcommand{\smallbmatrix}[1]{{\begin{bsmallmatrix}#1\end{bsmallmatrix}}}

\DeclareSymbolFont{bbold}{U}{bbold}{m}{n}
\DeclareSymbolFontAlphabet{\mathbbold}{bbold}
\newcommand{\onev}{\mathbbold{1}}
\newcommand{\onevtop}{\onev^{\!\!\top}\!}

\usepackage{calrsfs}
\DeclareMathAlphabet{\pazocal}{OMS}{zplm}{m}{n}
\renewcommand{\mathcal}[1]{\pazocal{#1}}

\usepackage[usenames]{xcolor}

\usepackage{mathtools}

\newcommand{\afterequation}{\vskip 3pt}

\newcommand{\mywidth}{.99\linewidth}

\begin{document}

\maketitle

\begin{abstract}
This paper introduces a theoretical framework for the
analysis and control of the stochastic susceptible-infected-removed
(SIR) spreading process over a network of heterogeneous agents.
In our analysis, we analyze the exact networked Markov process describing the SIR model, without resorting to mean-field approximations, and introduce a convex optimization framework to find an efficient allocation of resources to contain the expected number of accumulated infections over time.
Numerical simulations are presented to illustrate the
effectiveness of the obtained results.
\end{abstract}

\section{Introduction}

The analysis of contagion processes in complex networks is one of the
central problems in network science and engineering, with applications
in a wide range of scenarios, such as epidemiology~\cite{Nowzari2016},
public health~\cite{Tizzoni2012}, and cyber-physical
systems~\cite{Kim2012a}. During the last decade, we have witnessed a
tremendous advance in this problem, including the relationship between
epidemic thresholds and network eigenvalues~\cite{VanMieghem2009a},
the connection between curing policies and the cut-width of the graph
\cite{Drakopoulos2016a}, the use of optimization tools to contain
epidemic outbreaks \cite{Wan2008,Preciado2013,Preciado2014}, as well
as new modeling frameworks for analysis of spreading processes over
multilayer~\cite{DarabiSahneh2013,Chen2014}, time-varying
\cite{Ogura2016}, and adaptive networks~\cite{Guo2013,Ogura2015d}.

Designing strategies to contain epidemic outbreaks in networks is of
great relevance in public health. In this context, the following
question is of particular interest: given a contact network and
resources that provide partial protection, how should one distribute
these resources throughout the networks in a cost-optimal manner to
contain the spread? This question has been addressed in several papers
by the control community (see \cite{Nowzari2016} and references
therein). Most existing results are based on the analysis of the
susceptible-infected-susceptible (SIS) spreading model, in which nodes
in the network can only be in two states: infected or healthy.
However, in many practical settings, nodes can also be immune to the
disease due to, for example, a previous exposition to the infectant.
The addition of this third state has a nontrivial effect on the
dynamics of the spread, which is commonly modeled using the
susceptible-infected-removed (SIR) epidemic model \cite{Nowzari2016},
in which a node recovers from the infection with acquired immunity.

Although we find a variety of studies on the SIR model
\cite{Moreno2002a,Barthelemy2005,Sharkey2008,Sharkey2011,Khouzani2011a,Eshghi2016}, most of them are based on a {mathematical technique called} mean-field approximation, in which one assumes independence of (potentially) dependent random variables. {Although this approximation significantly simplifies the analysis of the model (as seen in the references above), the approximation is also known to introduce mathematical terms irrelevant to the original model and, therefore, can result in a large approximation error~\cite{Sharkey2011}.} The main goal of this paper is to introduce a theoretical framework for the analysis of the \emph{exact} networked Markov process describing the SIR model, without resorting to mean-field approximations, and introduce an optimization framework to find an efficient{, yet sub-optimal,} allocation of resources to contain the expected number of accumulated infections over time. Our framework extends to a generalized SIR model where an infected node can be isolated by authorities (e.g., quarantine) for the purpose of suppressing an epidemic outbreak. In this case, we present an alternative optimization framework to distribute a finite amount of resources to suppress an epidemic outbreak.

This paper is organized as follows. After introducing mathematical
preliminaries in Section~\ref{sec:SIR}, we describe the networked SIR
model (with and without isolated nodes), and state the resource
allocation problems analyzed in this paper. In Sections~\ref{sec:med}
and \ref{sec:isol}, we introduce a convex optimization framework to
provide solutions for these resource allocation problems. We
illustrate the effectiveness of our results via numerical simulations
in Section~\ref{sec:num}.

\subsection{Mathematical Preliminaries}

An undirected graph is a pair~$\mathcal G = (\mathcal V, \mathcal E)$, where
$\mathcal V = \{1, \dotsc, n\}$ is the set of nodes, and $\mathcal E$ is the set
of edges, consisting of distinct and unordered pairs~$\{i, j\}$ for $i, j\in
\mathcal V$. We say that $i$ and $j$ are adjacent if $\{i, j\} \in \mathcal E$.
The adjacency matrix~$A\in \mathbb{R}^{n\times n}$ of~$\mathcal G$ is defined as
the $\{0, 1\}$\nobreakdash-matrix whose $(i,j)$ entry is one if and only if $i$
and $j$ are adjacent.

For a positive integer $n$, define $[n] = \{1, \dotsc, n\}$. We let $\Id_n$
denote the identity matrix with dimension $n$ and $O_{n,m}$ denote the $n\times
m$ zero matrix. Let $u_i$ denote the $i$th canonical basis vector in $\mathbb{R}^p$ and
define $U_{ij} = u_i u_j^\top$. By $\onev_p$ we denote the $p$-vector whose
entries are all one. A real matrix $A$, or a vector as its special case, is said
to be nonnegative (positive), denoted by $A\geq 0$ ($A>0$, respectively), if $A$
is nonnegative (positive, respectively) entry-wise. We write $A\geq B$ if
$A-B\geq 0$. The notations $A>B$, $A\leq B$, and $A<B$ are then defined in the
obvious way. We denote the Kronecker product of matrices~$A$ and $B$ by
$A\otimes B$. We say that a square matrix is Hurwitz stable if all the
eigenvalues of the matrix have negative real parts. Also, we say that a square
matrix is Metzler if its off-diagonal entries are all non-negative.

For the proof of the main results of this paper, we need the following 
lemma {that equivalently reduces vector inequalities involving the inverse of a Metzler matrix to a pair of linear vector inequalities:}

\begin{lemma}[{\cite[Lemma~1]{Briat2012c}}]\label{lem:fundamental}
For a real number $\lambda$, a Metzler matrix $F \in \mathbb{R}^{n\times n}$ and nonnegative matrices $G
\in \mathbb{R}^{n\times s}$ and $H\in\mathbb{R}^{r\times n}$, the following
statements are equivalent:
\begin{itemize}
\item $F$ is Hurwitz stable and  $-\onev^{\!\top}_{r}HF^{-1}G
 < \lambda \onev_s^\top$;

\item there exists a positive vector $v \in \mathbb{R}^n$ satisfying the
inequalities $v^{\!\top} F + \onev_r^{\!\top} H < 0$ and $v^{\!\top} G < 
\lambda \onev_s^{\!\top}$.
\end{itemize}
\afterequation
\end{lemma}

Finally, we recall basic facts about a class of optimization problems called geometric programs~\cite{Boyd2007}. Let $x_1$, $\dotsc$, $x_m$ denote $m$ real positive variables. We say that a real-valued function $f$ of~$x = (x_1, \dotsc, x_m)$ is a {\it monomial function} if there exist $c>0$ and $a_1, \dotsc, a_m \in \mathbb{R}$ such that $f(x) = c {\mathstrut x}_1^{a_1} \dotsm {\mathstrut x}_m^{a_m}$. Also, we say that $f$ is a {\it posynomial function} if it is a sum of monomial functions of~$x$. Given posynomial functions $f_0$, $\dotsc$, $f_p$ and monomial functions $g_1$, $\dotsc$, $g_q$, the optimization problem
\begin{equation}\label{eq:generalGP}
\begin{aligned}
\minimize_x\ 
&
f_0(x)
\\
\subjectto\ 
&
f_i(x)\leq 1,\quad i=1, \dotsc, p, 
\\
&
g_j(x) = 1,\quad j=1, \dotsc, q, 
\end{aligned}
\end{equation}
is called a {\it geometric program}. It is known~\cite{Boyd2007} that a
geometric program can be converted into a convex optimization problem. We call
the constraints in \eqref{eq:generalGP} as {\it posynomial constraints}.

\section{SIR Model over Complex Networks} \label{sec:SIR}

In this section, we first give a brief overview of the networked SIR
(susceptible-infected-removed) model (see, e.g.,~\cite{Nowzari2016}).
We also present an extended version of the networked SIR model, where the
isolation of infected nodes is taken into account. We also state two resource
allocation problems under study.

\begin{figure}[tb]
\centering
\vspace{.2cm}
\includegraphics[width=.5\linewidth]{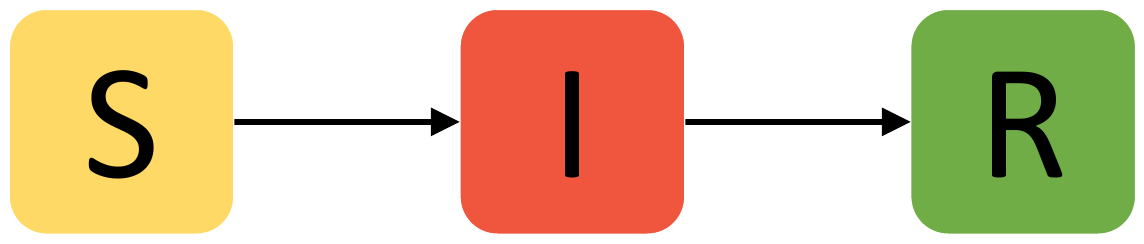}
\put(-90,18){$\beta_i$}
\put(-41,18){$\delta_i$}
\caption{SIR model with infection rates $\beta_i$ and recovery rates $\delta_i$.}
\label{fig:SIR}
\end{figure}

\subsection{SIR model}

Let $\mathcal G = (\mathcal V, \mathcal E)$ be an undirected graph of $n$ nodes
whose adjacency matrix equals $A$. In the networked SIR model, each node in the
graph can be in one out of three states: \emph{susceptible}, \emph{infected}, or
\emph{removed}. For convenience of notation, we represent the state of node~$i$
at time $t$ by the $\{0, 1\}$\nobreakdash-variables $S_i(t)$, $I_i(t)$, and $R_i(t)$ that
take value one if and only if node~$i$ is susceptible, infected, or removed,
respectively. Then, the dynamics of the state of the nodes is described by the
following transition probabilities
\begin{align}
\Pr(I_i(t+h) = 1 \mid S_i(t) = 1) &=
\beta_i\sum_{j=1}^n a_{ij}I_j(t) h + o(h),
\label{eq:infection}
\\
\Pr(R_i(t+h) = 1 \mid I_i(t) = 1) &= \delta_i \,h + o(h), 
\label{eq:recovery}
\end{align}
where $t \geq 0$ and $h > 0$ are arbitrary. The constants~$\beta_i > 0$ and
$\delta_i > 0$ are called the \emph{infection} and \emph{recovery} rate of
node~$i$, respectively (see Fig.~\ref{fig:SIR} for a schematic picture). The
probability~\eqref{eq:infection} indicates that a node~$i$ receives an infection
from each of its infected neighbors with the instantaneous rate of~$\beta_i$.
Also, from the latter probability~\eqref{eq:recovery}, the time it takes for
node~$i$ to recover from an infection event follows the exponential distribution
of mean~$1/\delta_i$. We assume that a node is either susceptible or infected at
time~$t=0$. We remark that, unlike in the SIS model widely studied in the
literature, we do not allow a transition from the infected or removed states to
the susceptible state, modeling the acquisition of immunity by nodes to
spreading process.

Let $\sigma_S(t)$,  $\sigma_I(t)$, and  $\sigma_R(t)$ denote the number of
susceptible, infected, and removed nodes at time $t$, respectively. In this
paper, we measure the prevalence of the spreading process by the quantity
\begin{equation}\label{def:lambda}
\lambda = \lim_{t\to\infty} E[\sigma_R(t)] - \sigma_I(0), 
\end{equation}
which equals the expected number of infections occurring after time $t=0$
because an infected node will be eventually removed with probability one.

We assume that, for each node $i$, we can distribute a preventative resource to
alter the value of~$\beta_i$ within a given interval~$[\ubar{\beta}_i,
\bar{\beta}_i] \subset (0, \infty)$ by paying a cost~$f_i(\beta_i)$. Similarly
we assume that we can distribute a corrective resource for tuning $\delta_i$
within another given interval~$[\ubar{\delta}_i, \bar{\delta}_i] \subset (0,
\infty)$ with an associated cost $g_i(\delta_i)$. Therefore, the total cost for achieving
the specific infection rates $\{\beta_1, \dotsc, \beta_n\}$ and the recovery
rates~$\{\delta_1, \dotsc, \delta_n\}$ equals $\sum_{i=1}^n (f_i(\delta_i) +
g_i(\beta_i))$. It is also assumed that we know which nodes in the graph are
infected at the initial time~$t = 0$.

\begin{figure}[tb]
\centering
\vspace{.2cm}
\includegraphics[width=.5\linewidth]{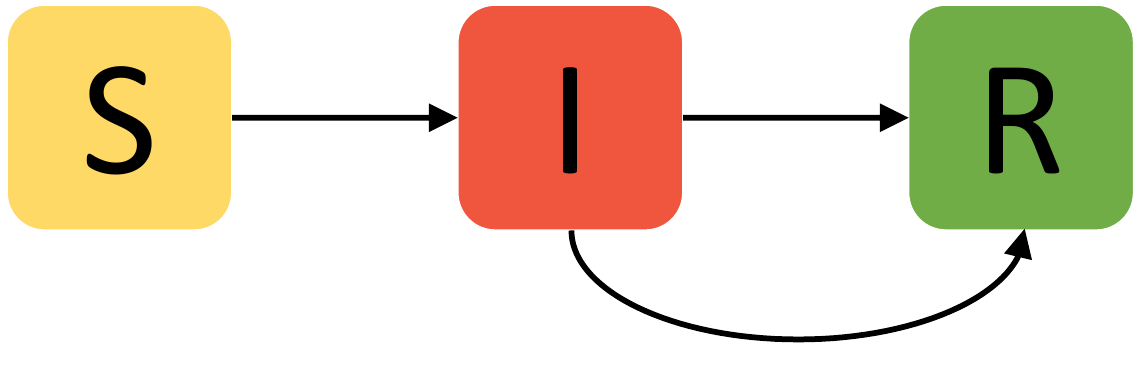}
\put(-90,31){$\beta_i$}
\put(-41,31){$\delta_i$}
\put(-41,6){$Y_i$}
\caption{SIR model with infection rates $\beta_i$, recovery rates $\delta_i$, and recovery times $Y_i$.}
\label{fig:SIR_isolation}
\end{figure}

Now we can state the first problem studied in this paper:

\begin{problem}[Resource allocation]\label{prb:medical}
Given a set of initially infected nodes, an available budget $\bar C > 0$, and a
desired control level $\bar\lambda > 0$, find $\beta_i \in [\ubar{\beta}_i,
\bar{\beta}_i]$ and $\delta_i \in [\ubar{\delta}_i, \bar{\delta}_i]$ ($i\in[n]$)
such that $\lambda \leq \bar\lambda$ and
\begin{equation}\label{eq:cost1<Gamma}
\sum_{i=1}^n (f_i(\delta_i) + g_i(\beta_i))\leq \bar C.
\end{equation}
\afterequation
\end{problem}

{We remark that, although the authors in
\cite{Preciado2014} give convex optimization-based solutions to the
optimal resource allocation problems for the SIS model, the allocation
optimal for the SIS model does not necessarily solve Problem~\ref{prb:medical}
for the SIR model, as we confirm in Section~\ref{sec:num}. Moreover,
since the infection-free equilibrium (i.e., the equilibrium where all
nodes are susceptible) is not the only equilibrium for the SIR model,
we cannot rely on the framework proposed in \cite{Preciado2014}.}

\subsection{SIR Model with Isolation}

The above problem statement is not practical when there are no medical resource able to tune  the recovery rate of infected patiences. An alternative action in this situation is to use social distancing, such as quarantine~\cite{Ki2015}. A typical question in this context is the following: \emph{How fast should we isolate infected nodes to effectively prevent an epidemic outbreak?}

In what follows, we introduce an extended SIR model that incorporates the isolation of
infected nodes. We assume that, once a node~$i$ becomes infected, the node is removed (quarantined) according to an stochastic process, which is independent of the
natural recovery process (as depicted in Fig.~\ref{fig:SIR_isolation}). Let $X_i$
denote the random variable following the exponential distribution with
mean~$1/\delta$ and $Y_i$ the time it takes to remove the node. Then, the
overall length of time from the infection of $i$ to its removal (either by
natural recovery or manual removal) is given by
\begin{equation}\label{eq:YwithIsolation}
Z_i = \min\left(X_i, Y_i\right).
\end{equation}

We assume that the distribution of the random variable~$Y_i$ is parametrized by
positive variables~$\gamma_i = (\gamma_{i1}, \dotsc, \gamma_{iq_i})$, and that
we can tune the distribution by paying a cost~$h_i(\gamma_i)$. We also assume
that, as in Problem~\ref{prb:medical}, we can tune the infection rates~$\beta_i$
with a cost~$g_i(\beta_i)$. We can then formulate the second problem studied
in this paper:

\begin{problem}[Resource allocation with isolation]\label{prb:isolation}
Given a set of initially infected nodes, an available budget $\bar C>0$, and a
desired control level $\bar\lambda > 0$, find $\beta_i \in [\ubar{\beta}_i,
\bar{\beta}_i]$ and $\gamma_i \in \prod_{k=1}^{q_i} [\ubar{\gamma}_{ik},
\bar{\gamma}_{ik}]$  ($i=1$, $\dotsc$, $n$) such that $\lambda \leq \bar\lambda$
and
\begin{equation}\label{eq:cost:isol}
\sum_{i=1}^n (g_i(\beta_i)+h_i(\gamma_i)) \leq \bar C.
\end{equation}
\end{problem}

\section{Resource Allocation without Isolation}\label{sec:med}

This section gives a solution to Problem~\ref{prb:medical}. We specifically show
that Problem~\ref{prb:medical} is feasible if a set of posynomial constraints is
feasible. For this purpose, we present an $n$\nobreakdash-states linear model that
upper-bounds the averaged dynamics of the \emph{exact} SIR model. We notice
that, since the SIR model is a Markov process having $3^n$ states, its direct
analysis is computationally hard, if not possible, even when the size~$n$ of the
graph is small.

We start with describing the SIR model by stochastic differential equations. Let
us denote by $N_\mu$ a Poisson counter of rate~$\mu$. We assume that all the
Poisson counters appearing in the paper are stochastically independent. Then,
from \eqref{eq:infection} and~\eqref{eq:recovery}, the evolution of the nodal
states can be exactly described by the stochastic differential equations:
\begin{equation} \label{eq:123.3hoge}
\begin{aligned}
dS_i &= -S_i\sum_{j=1}^n a_{ij} I_j \,dN_{\beta_i}, 
\\
dI_i &= -I_i \,dN_{\delta_i} + S_i \sum_{j=1}^n a_{ij} I_j \,dN_{\beta_i},
\\
dR_i &= I_i \,dN_{\delta_i}, 
\end{aligned}
\end{equation}
for $i=1$, $\dotsc$, $n$. Using these representations, we can
prove the following proposition: 

\begin{proposition}\label{prop:med}
Define the $n\times n$ diagonal matrices $J$, $B$, and $D$ by 
\begin{equation*}
J_{ii} = S_i(0),\ B_{ii} = \beta_i,\ D_{ii} = \delta_i
\end{equation*}
for each $i\in [n]$. {Let $\lambda$ be given by \eqref{def:lambda}
and} let $\bar\lambda > 0$ be an arbitrary vector. Then, we have
$\lambda < \bar\lambda$ if there exists a positive $v \in
\mathbb{R}^n$ satisfying the following inequalities:
\begin{subequations}\label{eq:linProg1}
\begin{gather}
v^\top J BA + \onev_n^\top D < v^\top D,\label{eq:linProg1.1}
\\
v^\top I(0) < \bar\lambda + \sigma_I(0). \label{eq:linProg1.2}
\end{gather}
\end{subequations}
\afterequation
\end{proposition}

\begin{proof}
From \eqref{eq:123.3hoge} it follows \cite{Sharkey2008} that\textcolor{white}{\eqref{eq:linProg1.2}}
\begin{equation}
\frac{d}{dt}E[I_i(t)] 
= 
-\delta_iE[I_i(t)] + \beta_i \sum_{j=1}^n a_{ij} E[I_j(t) S_i(t)], \label{eq:d/dtE[I_i(t)]}
\end{equation}
and $(d/dt)E[R_i(t)] = \delta_i E[I_i(t)]$. From the latter equation, the $\{0,
1\}^n$-valued stochastic processes $R = [R_1 \ \cdots \ R_n]^\top$ and $I = [I_1
\ \cdots \ I_n]^\top$ satisfy $(d/dt)E[R(t)] = D E[I(t)]$. Therefore,
\begin{equation}\label{eq:lim...}
\lim_{t\to\infty} E[R(t)] = D \int_0^\infty E[I(t)]\,dt. 
\end{equation}

Let us evaluate the integral~$\int_0^\infty E[I(t)]\,dt$. Since $S_i$ is
non-increasing as a function of $t$, we have $S_i(t) \leq S_i(0) = J_{ii}$.
Therefore, by \eqref{eq:d/dtE[I_i(t)]} we have 
\begin{equation*}
\frac{d}{dt}E[I_i(t)]  \leq
-\delta_iE[I_i(t)] + J_{ii} \beta_i \sum_{j=1}^n a_{ij} E[I_j(t)]
\end{equation*} 
for every $i$, Therefore, we obtain
\begin{equation*}
\frac{d}{dt}E[I(t)] \leq (JBA-D) E[I(t)]
\end{equation*}
and, hence,
\begin{equation}\label{eq:ddtI<...}
E[I(t)] \leq \exp\left(
(JBA-D)t
\right)I(0)
\end{equation}
by the comparison principle~\cite{Kirkilionis2004}. Since $JBA-D$ is Hurwitz
stable by \eqref{eq:linProg1.1} and Lemma~\ref{lem:fundamental}, we can integrate
\eqref{eq:ddtI<...} to obtain
\begin{equation*}
\begin{aligned}
\int_0^\infty E[I(t)]\,dt 
&\leq \int_0^\infty
\exp\left((JBA-D)t\right)I(0)\,dt 
\\
&= -(JBA-D)^{-1} I(0).
\end{aligned}
\end{equation*}
From this inequality and \eqref{eq:lim...}, we see that
\begin{equation*}
\begin{aligned}
\lambda 
&= 
\onev_{n}^\top \lim_{t\to\infty}
E[R(t)] - \sigma_I(0)
\\
&\leq -\onev_{n}^\top D (JBA-D)^{-1} I(0) -\sigma_I(0). 
\end{aligned}
\end{equation*}
Therefore, by Lemma~\ref{lem:fundamental}, we can conclude $\lambda \leq
\bar\lambda$ from the existence of a positive vector $v \in\mathbb{R}^n$
satisfying \eqref{eq:linProg1}. This completes the proof of the proposition.
\end{proof}

The inequalities in \eqref{eq:linProg1} allow us to achieve an efficient
resource distribution via checking the feasibility of posynomial constraints,
under the assumption that the cost functions~$f_i$ and~$g_i$ are posynomials for
every $i\in [n]$:

\begin{theorem}\label{thm:optimal}
Assume that $f_i$ and $g_i$ are posynomials for every $i\in [n]$. Then,
$\{\beta_i\}_{i=1}^n$ and $\{\delta_i\}_{i=1}^n$ solve Problem~\ref{prb:medical}
if there exists a positive $v\in\mathbb{R}^n$ satisfying the following
posynomial constraints:
\begin{equation}\label{eq:GP:med}
\text{\eqref{eq:cost1<Gamma}, 
\eqref{eq:linProg1}, 
$\ubar{\beta}_i \leq \beta_i \leq \bar{\beta}_i$, 
and $\ubar{\delta}_i \leq \delta_i \leq \bar{\delta}_i$.
}
\end{equation}
\afterequation
\end{theorem}

\begin{proof}
It is easy to check that all the constraints in \eqref{eq:GP:med} are posynomial
constraints, provided $f_i$ and $g_i$ are posynomials for every $i$. Moreover,
from Proposition~\ref{prop:med}, we can easily see that the rates $\beta_i$ and
$\delta_i$ satisfying the constraints \eqref{eq:GP:med} solve
Problem~\ref{prb:medical}.
\end{proof}

\section{Resource Allocation with Isolation} \label{sec:isol}

This section presents a solution to Problem~\ref{prb:isolation}. As in the
previous section, we show that Problem~\ref{prb:isolation} is feasible if a set
of posynomial constraints is feasible, under the general assumption that the
removal times~$Y_i$ can be described by a general class of distributions called
phase-type distributions~\cite{Asmussen1996}.

We start our presentation by reviewing phase-type
distributions~\cite{Asmussen1996}. Consider a time-homogeneous Markov process in
continuous-time with $p+1$ ($p\geq 1$) states such that states~$1$, $\dotsc$,
$p$ are transient and state $p+1$ is absorbing. The infinitesimal generator of
the process, that is, the matrix of transition rates, is then necessarily of the
form
\begin{equation*}
\begin{bmatrix}
\Pi & w\\
0 & 0
\end{bmatrix}, \ w = -\Pi \onev_p,
\end{equation*}
where $\Pi \in \mathbb{R}^{p\times p}$ is an invertible Metzler matrix with
non-positive row-sums. Let $\smallbmatrix{\phi\\0} \in \mathbb{R}^{p+1}$ ($\phi
\in \mathbb{R}^p$) denote the initial distribution of the
Markov process. Then, the time to absorption into the state~$p+1$, denoted by
$(\phi, \Pi)$, is called a \emph{phase-type distribution}. It is known that the
set of phase-type distributions is dense in the set of positive valued
distributions~\cite{Asmussen1996}, even when $\phi = u_1$. Moreover, there
exists an efficient fitting algorithm to approximate a given arbitrary
distribution by a phase-type distribution \cite{Asmussen1996}.

We can now state our assumption on the distribution of the removal times~$Y_i$:

\begin{assumption}
There exists $p\geq 0$ such that, for each $i\in [n]$, $Y_i$ follows the
phase-type distribution $(u_1, \Pi_i)$, with $\Pi_i \in \mathbb{R}^{p\times p}$
being parametrized by positive variables~$\gamma_i = (\gamma_{i1}, \dotsc,
\gamma_{iq_i})$.
\end{assumption}

The rest of this section is devoted to giving a solution to
Problem~\ref{prb:isolation}. We start with showing that the overall removal
time~$Z_i$ given in \eqref{eq:YwithIsolation} follows a phase-type distribution:

\begin{lemma}\label{lem:PH}
Define $\Pi'_i = \Pi_i-\delta_i \Id_p$. Then $Z_i$ follows the phase-type
distribution $(u_1, \Pi'_i)$.
\end{lemma}

\begin{proof}
We first recall that the cumulative distribution function of the phase-type
distribution $(\phi, \Pi)$ equals $F(t) = 1- \phi \exp(t\Pi)\onev_p$. We also
recall that, given independent random variables $X_1$ and $X_2$ having the
cumulative distribution functions $F_1$ and $F_2$, the cumulative distribution
function of the random variable $\min(X_1, X_2)$ equals
$1-(1-F_1(t))(1-F_2(t))$. From the above facts, the cumulative distribution
function of~$Z_i$ equals
\begin{equation*}
\begin{aligned}
F(t) 
&= 
1-\bigl(1 - (1-e^{-\delta_i t})
\bigr)\bigl(1 - (1 -u_1\exp(t\Pi_i)\onev_p )\bigr)
\\
&=
1- e^{-\delta_i t} u_1\exp(t\Pi_i)\onev_p
\\
&=
1 - u_1 \exp\bigl(t(\Pi_i-\delta_i \Id_p)\bigr) \onev_p, 
\end{aligned}
\end{equation*}
which coincides with the cumulative distribution function of~$(u_1, \Pi'_i)$, as
desired.
\end{proof}

From Lemma~\ref{lem:PH} we see that, in our SIR model with isolations, the
infection rates are constants and the overall removal times~$Z_i$ follow
phase-type distributions. Therefore, this SIR model has a similar structure with
the SIS model studied in~\cite{Ogura2015b}, where it is assumed that the
infection rates are constants while the recovery (i.e., the transition from the
infected state to the susceptible state) occurs following phase-type
distributions. Hence, following the same argument as in~\cite{Ogura2015b}, we
can describe our SIR model with isolations using stochastic differential
equations:

\begin{proposition}\label{prop:SDEforSIRisolation}
Define $w'_i = [w'_{i1}\ \cdots \ w'_{ip}]^\top = -\Pi'_i \onev_p$. Let the
$\{0, 1\}$-valued stochastic processes $S_i, R_i$ and the $\{0, 1\}^p$\nobreakdash-valued
stochastic processes $\tilde I_{i} = [\tilde I_{i,1}\ \cdots\ \tilde
I_{i,p}]^\top$ ($i=1, \dotsc, n$) follow the stochastic differential equations:
\begin{align}
dS_i 
&=
-S_i\sum_{j=1}^n a_{ij} \onev_p^\top \tilde I_j\, dN_{\beta_i},  \notag
\\
d\tilde I_i
&= 
\sum_{\ell, m=1}^p(U_{m\ell}-U_{\ell\ell}) \tilde I_{i}\,dN_{\Pi'_{i,\ell m}}
- 
\sum_{\ell=1}^p U_{\ell \ell}\tilde I_{\ell}\,dN_{w'_{i\ell}}\notag
\\
&\hspace{3cm}+ 
u_1S_i\sum_{j=1}^n a_{ij} \onevtop \tilde I_j\, dN_{\beta_i}, 
\label{eq:isol:dI_i}
\\
dR_i &= \sum_{\ell=1}^p\tilde  I_{i,\ell} \,dN_{w'_{i\ell}}, \label{eq:isol:dR_i}
\end{align}
with the initial conditions
\begin{equation*}
\begin{multlined}[.85\linewidth]
\bigl(S_i(0), \tilde I_i(0), R_i(0) \bigr)  \\= 
\begin{cases}
(0, u_1, 0),&\text{if $i$ is infected at time $0$},
\\
(1, 0, 0),&\text{otherwise}. 
\end{cases}
\end{multlined}
\end{equation*}
Define 
\begin{equation*}
I_i(t) = \onev_p\tilde I_i(t).
\end{equation*}
Then, in the SIR model with isolations, a node $i$ is susceptible, infected, or
removed at time $t$ if and only if $S_i(t)=1$, $I_i(t)=1$, or $R_i(t)=1$,
respectively.
\end{proposition}

\begin{proof}
We refer the readers to the proof of \cite[Proposition~3.1]{Ogura2015b}. The
details are omitted.
\end{proof}

Based on Proposition~\ref{prop:SDEforSIRisolation}, we can prove the following
proposition:

\begin{proposition} \label{prop:isol}
Let $\bar \lambda >0$ be given.  Then, we have $\lambda < \bar\lambda$ if there
exists a positive vector $v \in \mathbb{R}^{np}$ satisfying the following
inequalities:
\begin{gather}
v^\top \biggl( \bigoplus_{i=1}^n (\Pi'_i)^\top + (JBA)\otimes (u_1\onev_p^\top)\biggr) + \onev_n^\top \bigoplus_{i=1}^n (w_i')^\top  < 0, \label{eq:isol:inequ}
\\
v^\top \tilde I(0)  < \bar\lambda + \sigma_I(0). \label{eq:isol:ineq2}
\end{gather}
\afterequation
\end{proposition}

\begin{proof}
Define $\tilde I = [\tilde I_1^{\,\top} \ \cdots \tilde I_n^{\,\top}]^\top$.
From \eqref{eq:isol:dR_i} it follows that 
\begin{equation*}
(d/dt)E[R_{i}(t)] 
= \sum_{\ell=1}^{p}
E[\tilde I_{i, \ell}(t)]w'_{i\ell} 
= (w'_i)^\top E[\tilde I_i(t)]
\end{equation*}
and therefore
$(d/dt)E[R(t)] = (\bigoplus_{i=1}^n (w'_i)^\top) E[\tilde I(t)].$ This equation
shows that
\begin{equation}\label{eq:limE[R(t)]=...}
\lim_{t\to\infty}E[R(t)] = \biggl(\bigoplus_{i=1}^n (w'_i)^\top \biggr) \int_0^\infty E[\tilde I(t)]\,dt. 
\end{equation}
On the other hand, in the same way as in~\cite{Ogura2015b}, from
\eqref{eq:isol:dI_i} we observe $({d}/{dt})E[\tilde I(t)] \leq \mathcal A
E[\tilde I(t)]$, where
\begin{equation*}
\mathcal A 
= 
\bigoplus_{i=1}^n (\Pi'_i)^\top + (JBA)\otimes (u_1\onev_p^\top).
\end{equation*}
Therefore, 
\begin{equation}\label{eq:E[I(t)]leq...isol}
E[\tilde I(t)] \leq \exp(\mathcal A t) \tilde I(0). 
\end{equation}
Notice that, by \eqref{eq:isol:inequ} and Lemma~\ref{lem:fundamental}, the
matrix $\mathcal A$ is Hurwitz stable. Therefore, in the same way as in the
proof of Proposition~\ref{prop:med}, from \eqref{eq:limE[R(t)]=...} and
\eqref{eq:E[I(t)]leq...isol} it follows that the condition stated in the theorem
is sufficient for $\lambda \leq \bar\lambda$. The details are omitted.
\end{proof}

Based on Proposition~\ref{prop:isol}, we can derive the following theorem, which enables us
to solve Problem~\ref{prb:isolation} by checking the feasibility of posynomial constraints:

\begin{theorem}
Let $\mathcal D\Pi_i$ and $\mathcal O \Pi_i$ denote the diagonal and
off-diagonal parts of $\Pi_i$, respectively. For all $i \in \{1, \dotsc, n\}$,
assume that
\begin{enumerate}
\item the entries of $w_i$ and $\mathcal O \Pi_i$ are posynomials in~$\gamma_i$;

\item the entries of $-\mathcal D\Pi_i$ are monomials in~$\gamma_i$; 

\item the cost functions $g_i$ and $h_i$ are posynomials in $\beta_i$ and~$\gamma_i$, respectively.
\end{enumerate}
Moreover, for all $i \in [n]$ and $\ell \in [p]$, let $\kappa_{i\ell}$ and
$\alpha_{i\ell}$ be positive constants such that
\begin{equation}\label{eq:mono<linear}
\kappa_{i\ell} (-\Pi_{i,\ell\ell})^{\alpha_{i\ell}} \leq (-\Pi_{i,\ell\ell}) + \delta_i
\end{equation}
for all possible values of $\gamma_i$. Then, the parameters $\beta_1$, $\dotsc$,
$\beta_n$, $\gamma_1$, $\dotsc$, $\gamma_n$ solve Problem~\ref{prb:isolation} if
there exists a positive vector~$v\in \mathbb{R}^{np}$ satisfying the posynomial
constraints:
\begin{subequations}\label{eq:GP:isoL}
\begin{gather}
\begin{aligned}
&v^\top \biggl(
\bigoplus_{i=1}^n (\mathcal O \Pi_i)^\top +  (JBA)\otimes (u_1 \onev_p^\top)\biggr)
+\onev_p^\top \bigoplus_{i=1}^n w_i^\top +\\
&\hspace{1.9cm}\onev_n^\top (D\otimes \onev_p)< v^\top \bigoplus_{i=1}^n\bigoplus_{\ell=1}^p \kappa_{i\ell}(-\Pi_{i,\ell\ell})^{\alpha_{i\ell}},
\end{aligned}\label{eq:GP:isol:main}
\\
\eqref{eq:cost:isol},\ 
\eqref{eq:isol:ineq2},\ 
\ubar{\beta}_i \leq \beta_i \leq \bar \beta_1,\  
\gamma_i \in \prod_{k=1}^{q_i} [\ubar{\gamma}_{ik},
\bar{\gamma}_{ik}].
\label{eq:GP:isol:constraints}
\end{gather}
\end{subequations}
\end{theorem}

\begin{proof}
It is straightforward to see that the conditions in~\eqref{eq:GP:isoL}
are all posynomial constraints with respect to variables~$v$,
$\beta_i$, and~$\gamma_i$ ($i\in [n]$) under the assumptions stated in
the theorem. Assume that $\beta_1$, $\dotsc$, $\beta_n$, $\gamma_1$,
$\dotsc$, $\gamma_n$ satisfy the constraints~\eqref{eq:GP:isol:main}
and \eqref{eq:GP:isol:constraints}. Then, by the conditions on the
constants $\kappa_\ell$ and $\alpha_\ell$, we have
$\bigoplus_{i=1}^n\bigoplus_{\ell=1}^p
\kappa_{i\ell}(-\Pi_{i,\ell\ell})^{\alpha_{i\ell}} \leq \delta
\Id_{np} + \bigoplus_{i=1}^n(-\mathcal D \Pi_i)$. Substituting this
inequality to \eqref{eq:GP:isol:main} indeed yields
\eqref{eq:isol:inequ}. This observation and the
constraints~\eqref{eq:GP:isol:constraints} guarantee $\lambda \leq
\bar\lambda$ by Proposition~\ref{prop:isol}.
\end{proof}

\begin{remark}
We can use a bisection search to find the pair $(\kappa_{i\ell},
\alpha_{i\ell})$ that satisfy \eqref{eq:mono<linear} and, moreover, minimize the
maximum difference between the left- and right-hand sides of
\eqref{eq:mono<linear}. The details are omitted due to limitations of space.
\end{remark}

\section{Numerical Examples}\label{sec:num}

We present numerical examples in this section, for both the cases without and
with isolations (i.e., quarantine). We let $\mathcal G$ be the graph of a part of a social network
of $n = 68$ nodes. The adjacency matrix of the graph has the spectral radius
$\rho = 10.61$. Also, we randomly choose and fix four initially infected nodes
from the graph.

We first consider the case without isolations. For each $i\in [n]$,
let $\ubar \beta_i = \ubar \beta = 0.00266$, $\bar \beta_i = \bar
\beta = 0.0133$, $\ubar \delta_i = \ubar \delta = 0.05$, and $\bar
\delta_i = \bar \delta = 0.1$. We use the cost functions $f_i(\beta_i)
= c_1 \beta_i^{-1} + c_2$ and $g_i(\delta_i) = c_3 \delta_i + c_4$,
where the real constants $c_1>0$, $c_2$, $c_3>0$, and $c_4$ are chosen
in such a way that $f_i(\ubar \delta) = 1$, $f_i(\bar \delta) = 0$,
$g_i(\ubar \beta) = 0$, and $g_i(\bar \beta) = 0$. We let $\bar C = n
= 68$. To find sub-optimal resource allocations, we minimize
$\bar\lambda$ subject to the constraints in \eqref{eq:GP:med}. This
optimization problem is a geometric program because $\bar \lambda$ is
trivially a monomial. The scatter plot of the costs from the obtained
recovery rates and infection rates is shown in
Fig.~\ref{fig:cost-cost} (red plots). We can observe an interesting
difference of the obtained cost allocation from the one based on the
SIS model~\cite{Preciado2014} (blue plots).  In
Fig.~\ref{fig:degree-cost}, we show the scatter plot of the
sub-optimal investments on each node versus degrees of the nodes.
Using Monte Carlo simulation, we find that the proposed allocation
achieves $\lambda = 2.57$, which is 40\% less than $\lambda = 4.38$
obtained from the allocation~\cite{Preciado2014} optimized for the SIS
model.

\begin{figure}[tb]
\vspace{.2cm}
\centering
\includegraphics[trim=2cm 0 2cm 0,width=\mywidth]{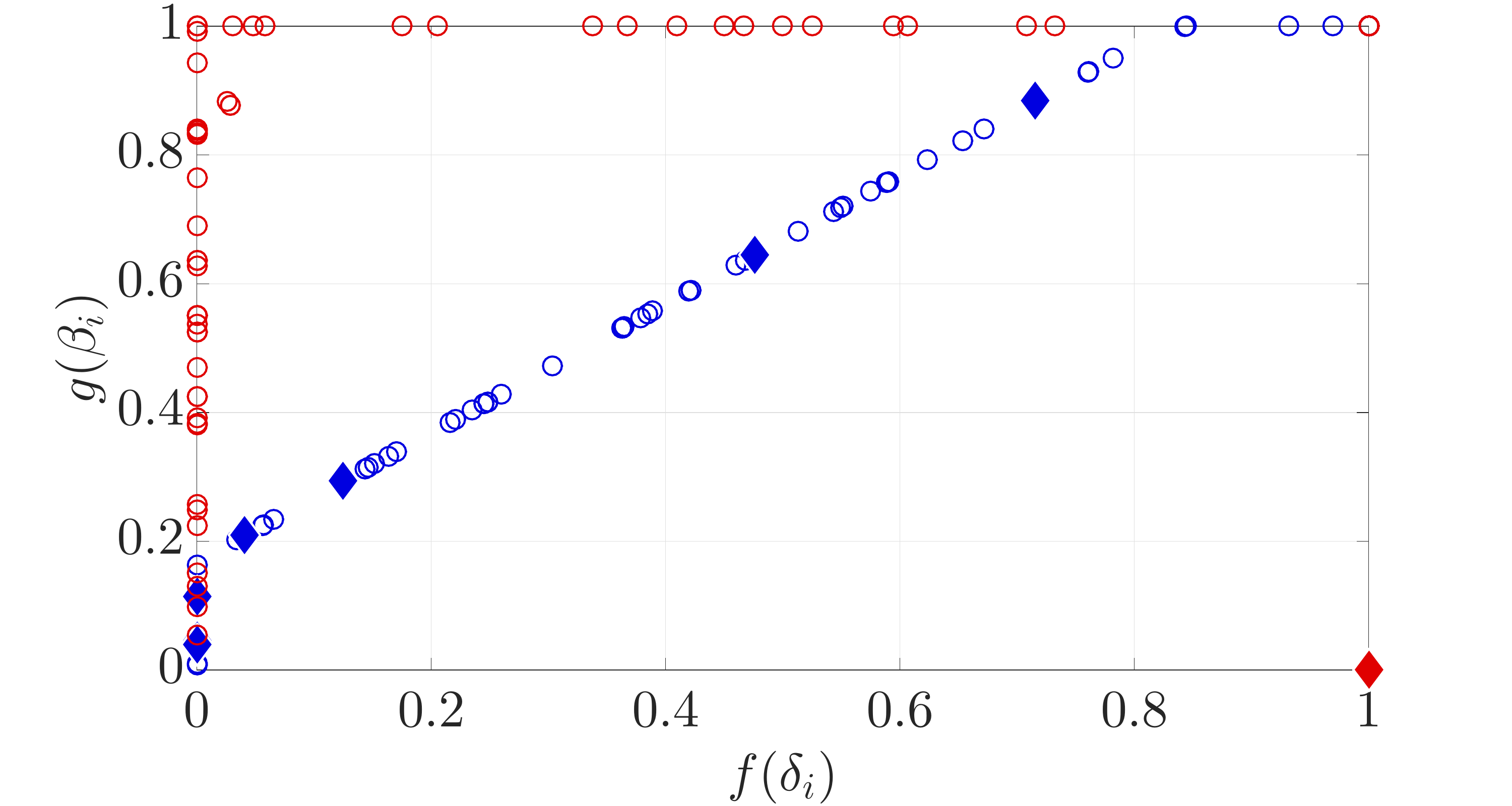}
\caption{Correction versus prevention per node. Red:  Proposed. Blue: The optimal allocation
for the SIS model~\cite{Preciado2014}.  Diamond markers correspond to the initially infected nodes.}
\label{fig:cost-cost}
\vspace{.5cm}
\centering
\includegraphics[trim=2cm 0 2cm 0,width=\mywidth]{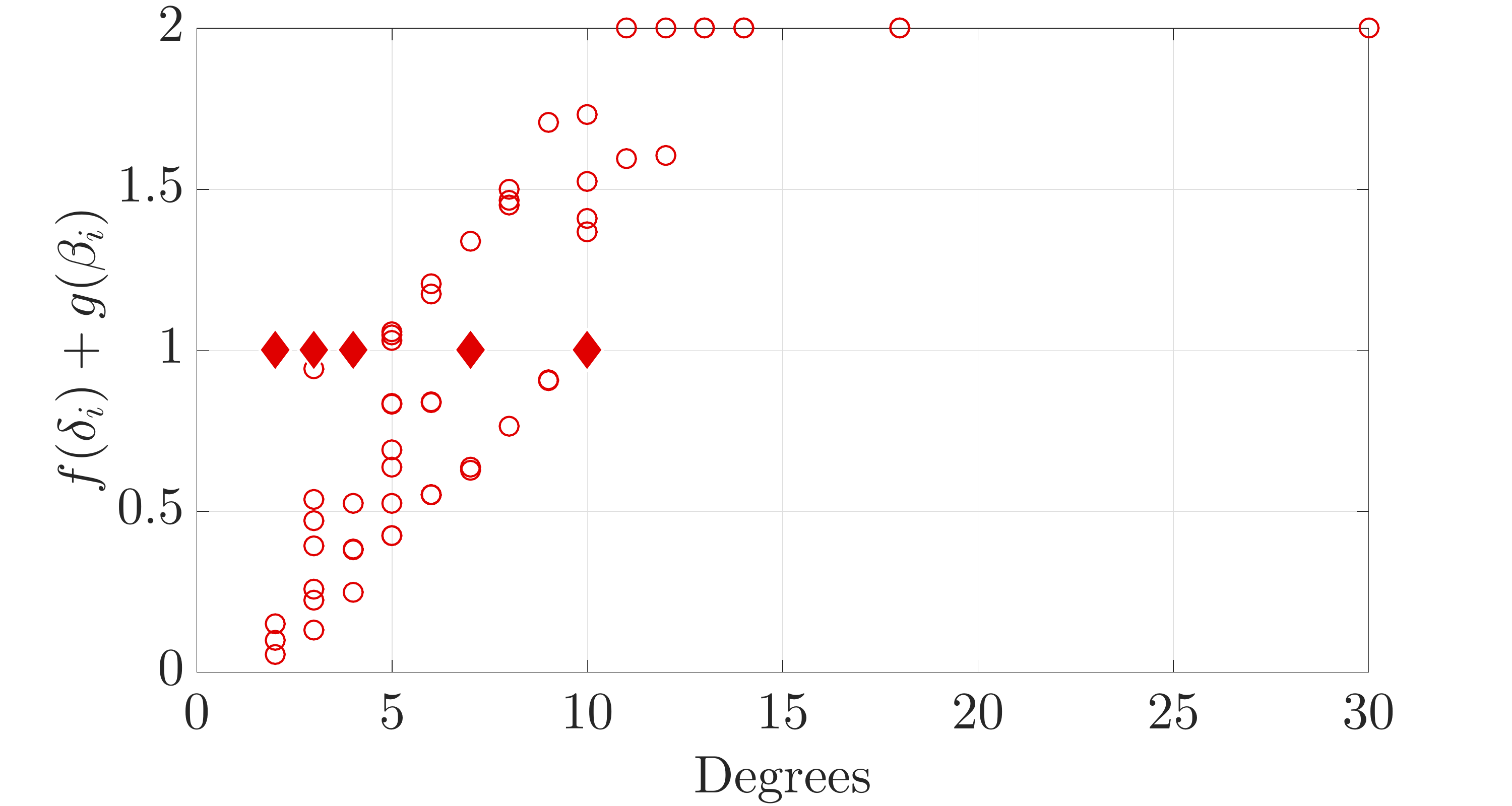}
\caption{Sub-optimal investments per node versus degrees. Diamond markers correspond to the initially infected nodes}
\label{fig:degree-cost}
\end{figure}

We then consider the case with isolations. We model the isolation times~$Y_i$ by
Erlang distributions with mean $\gamma_i > 0$ and shape $p$. We let $\gamma_i$
be the design variable. Since an Erlang distribution is a $p$-sum of independent
and identically distributed exponential distributions, $Y_i$ approximates a
normal distribution when $p$ is large. Moreover, $Y_i$ is the phase type
distribution $(e_1, \Pi_i)$ with
\begin{equation*}
\Pi_i = \begin{bmatrix}
	-p/\gamma_i & p/\gamma_i  &        &  O\\
	          & \ddots & \ddots &  \\
	          &           & \ddots & p/\gamma_i  \\
{O}	          &           &        & -p/\gamma_i
\end{bmatrix} \in \mathbb{R}^{p\times p}. 
\end{equation*}
Since
\begin{equation*}
w_i = \begin{bmatrix}
O_{p-1,1}\\p/\gamma_i
\end{bmatrix}, \ 
\mathcal O\Pi_i=
\begin{bmatrix}
O_{p-1,1}&(p/\gamma_i)\Id_{p-1}
\\
0 & O_{1, p-1}
\end{bmatrix},
\end{equation*}
and $-\mathcal D \Pi_i = (p/\gamma_i)\Id_{p}$, all the assumptions in
Theorem~\ref{thm:optimal} are satisfied. We choose the cost function for
$\gamma_i$ as $h_i(\gamma_i) = c_5 /\gamma_i + c_6$, where $c_5>0$ and $c_6$ are
constants such that $h(\ubar{\gamma}_i) = 1$ and $h(\bar{\gamma}_i) = 0$. This
choice is based on an assumption that we have to pay the more cost to  achieve
the faster response to patients. We  fix $\delta_i = 0.1$ and $\bar C = n$. To
find sub-optimal resource allocations, we minimize $\bar\lambda$ subject to the
constraints in \eqref{eq:GP:isoL}. We show the scatter plot of the sub-optimal
resource allocation in Fig.~\ref{fig:cost-cost:isolation}, where we can observe
a similar pattern as Fig.~\ref{fig:cost-cost}. 

\begin{figure}[tb]
\vspace{2mm}
\centering
\includegraphics[trim=2cm 0 2cm 0,width=\mywidth]{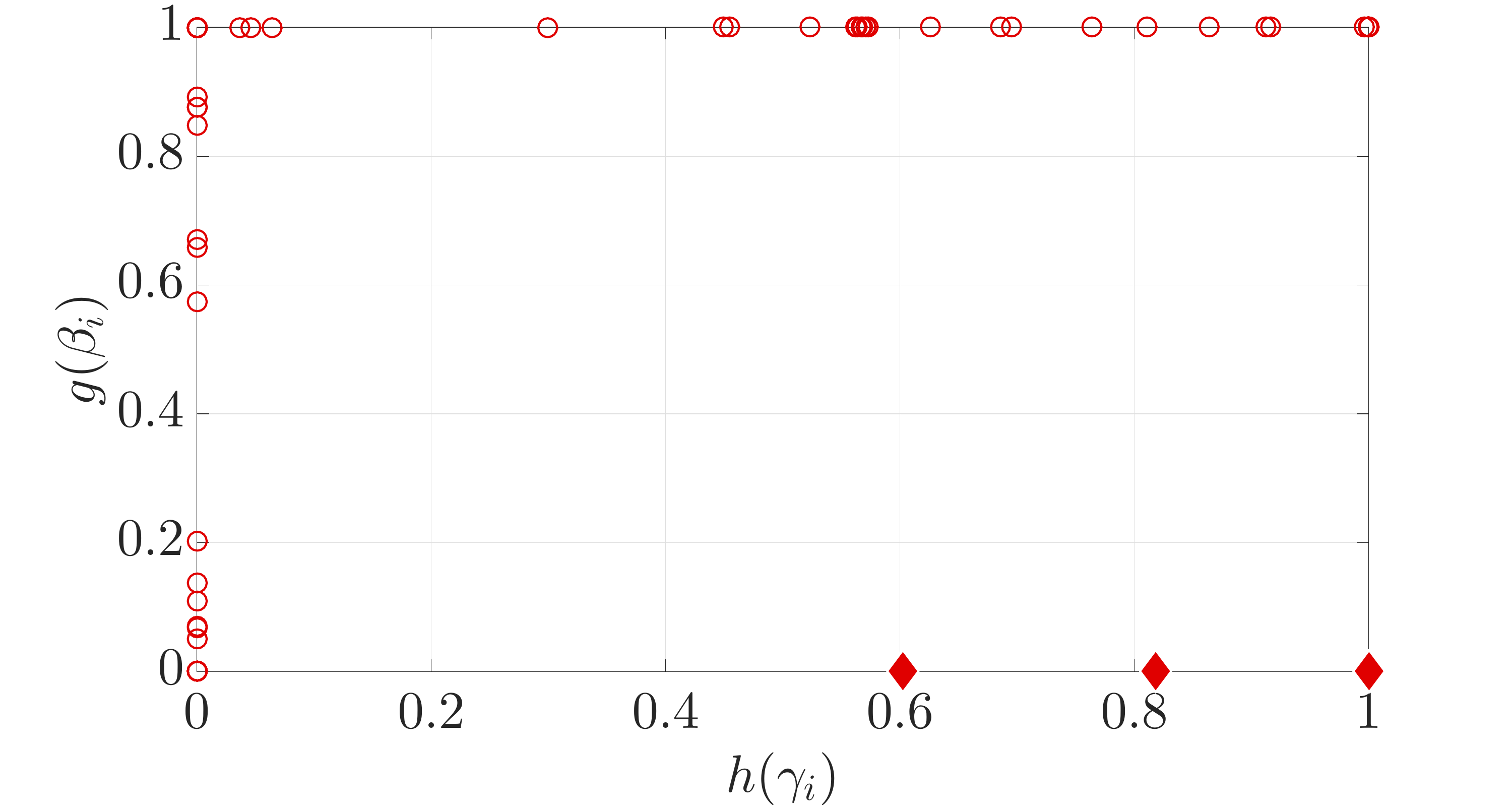} \caption{Isolation versus prevention per node. Diamond markers correspond to the initially infected nodes.} \label{fig:cost-cost:isolation}
\vspace{-3mm}
\end{figure}

\section{Conclusion}

In this paper, we have proposed a convex optimization framework to contain an epidemic outbreak in the networked SIR models. We have developed a framework to find a sub-optimal resource allocation to contain the accumulated number of
infections over time. We have then extended our results to a networked SIR model allowing isolations (quarantines), where infected nodes can be removed from the population.
We have then illustrated the efficiency of our framework via numerical simulations in a real social network.




\end{document}